\documentclass{llncs}
\usepackage[T1]{fontenc}
\usepackage{amsmath, amssymb}
\usepackage{mathtools}
\usepackage{comment}

\usepackage{graphicx}
\usepackage{booktabs}
\usepackage{threeparttable}
\usepackage{makecell}
\usepackage{array}
\usepackage{pifont}

\usepackage{amssymb}    
\usepackage{centernot}  

\usepackage{booktabs}
\usepackage{tabularx}
\usepackage{threeparttable}
\usepackage{pifont}

\newcommand{\cmark}{\ding{51}}
\newcommand{\xmark}{\ding{55}}

\usepackage{booktabs}
\usepackage{tabularx}
\usepackage{threeparttable}
\usepackage{makecell}
\usepackage{ragged2e}
\usepackage{array}
\usepackage{pifont}

\newcolumntype{Y}{>{\RaggedRight\arraybackslash}X}
\newcolumntype{L}[1]{>{\RaggedRight\arraybackslash}p{#1}}
\newcolumntype{L}[1]{>{\raggedright\arraybackslash}p{#1}}

\usepackage{tikz}
\usetikzlibrary{fit,backgrounds, arrows.meta, calc, positioning, matrix}

\usepackage{algorithm}
\usepackage{algpseudocode}
\usepackage{placeins} 
\usepackage{url}
\newcommand{\doi}[1]{\url{https://doi.org/#1}}

\spnewtheorem{claimnum}[theorem]{Claim}{\bfseries}{\itshape}
\renewenvironment{claim}[1][]{\begin{claimnum}[#1]}{\end{claimnum}}

\spnewtheorem{lemmanum}[theorem]{Lemma}{\bfseries}{\itshape}
\renewenvironment{lemma}[1][]{\begin{lemmanum}[#1]}{\end{lemmanum}}

\spnewtheorem{propnum}[theorem]{Proposition}{\bfseries}{\itshape}
\renewenvironment{proposition}[1][]{\begin{propnum}[#1]}{\end{propnum}}

\spnewtheorem{coronum}[theorem]{Corollary}{\bfseries}{\itshape}
\renewenvironment{corollary}[1][]{\begin{coronum}[#1]}{\end{coronum}}

\spnewtheorem{observation}[theorem]{Observation}{\bfseries}{\itshape}

\usepackage{booktabs}
\usepackage{tabularx}
\usepackage{threeparttable}
\usepackage{ragged2e}
\usepackage{array}
\usepackage{pifont}

 \newcolumntype{Y}{>{\RaggedRight\arraybackslash}X}
\newcolumntype{L}[1]{>{\RaggedRight\arraybackslash}p{#1}}
\newcolumntype{C}[1]{>{\Centering\arraybackslash}p{#1}}

\providecommand{\qedsymbol}{\ensuremath{\square}}
\providecommand{\qedtrianglesymbol}{\ensuremath{\triangle}} 

\newcommand{\qedbox}{\hfill\qedsymbol}
\newcommand{\qedtri}{\hfill\qedtrianglesymbol}

\usepackage{bbding}

\author{Samuel German \Envelope}
\authorrunning{Samuel German}
\institute{University of California, San Diego, USA \\
\email{sgerman@ucsd.edu}}

\title{Layer-Based Width for \textsc{PAFP}}

\begin{document}
\maketitle

\begin{abstract}
The Path Avoiding Forbidden Pairs problem (\textsc{PAFP}) asks whether, in a
directed graph $G$ with two of its vertices labeled  $s,t$ and a set $\mathcal F$ of forbidden
vertex pairs, there is an $s$--$t$ path that contains at most one endpoint
from each pair in $\mathcal F$. We initiate the study of \textsc{PAFP} through a 
layer-based width measure. We first look at the directed union digraph
$G\cup\mathcal F$, which is obtained by adding one arc per forbidden pair to $G$, oriented
according to a fixed reachability-compatible order. Let the BFS layer $L_d$ be the set of vertices at directed shortest-path distance $d$ from $s$, and define the BFS-width from $s$ as $\max_d |L_d|$. We show that if $G\cup\mathcal F$ has BFS-width $b$ from $s$, and
$p$ is the maximum, over BFS-layer cuts, of the smaller of the numbers
of distinct heads and distinct tails of backward arcs crossing the cut,
then \textsc{PAFP} is FPT parameterized by $b+p$.
Here a backward arc goes from a later BFS layer to an earlier one.
This permits an unbounded total number of backward arcs.
In contrast, we show that bounded BFS-width alone does not suffice: \textsc{PAFP}
remains $\mathsf{NP}$-complete when the union digraph has BFS-width 2.

We also study layer-based width parameters on the input graph.  We show that if the input DAG has
BFS-width at most $2$ and only $k$ backward input arcs, then \textsc{PAFP}
can be decided in $2^k |I|^{O(1)}$ time, with unrestricted forbidden pairs.
This width-$2$ result is tight, in the sense that it follows from inspection of a
classical reduction that we have
$\mathsf{NP}$-completeness of \textsc{PAFP} on input DAGs of BFS-width $3$ with no
backward input arcs. 

Moreover, we study exact-length layers in the input graph, where the $d$-th exact-length layer consists of the
vertices reachable from $s$ by a directed path of length exactly $d$. We show that for
DAGs of exact-length width at most $2$, \textsc{PAFP} is polynomial-time
decidable by a $2$-SAT encoding of fixed-length paths. This bound is tight in that
the aforementioned reduction gives $\mathsf{NP}$-completeness on DAGs
of exact-length width $3$. Unlike previously known polynomial-time regimes
for \textsc{PAFP}, which restrict the forbidden-pair set in order to obtain tractability, our 2 input-graph tractability results allow unrestricted forbidden
pairs and input graphs with exponentially many $s$--$t$ paths.

\end{abstract}

\section{Introduction}
The \emph{Path Avoiding Forbidden Pairs} problem (\textsc{PAFP}) asks, given a directed graph $G=(V,E)$,
terminals $s,t\in V$, and forbidden vertex pairs $\mathcal F\subseteq \binom{V}{2}$, whether there exists an
$s$--$t$ path that contains at most one endpoint from each pair in~$\mathcal F$. \textsc{PAFP} is a natural abstraction for path selection under pairwise incompatibility constraints. It first arose in automatic software testing: a program is represented by a directed graph whose vertices represent code segments and whose edges represent the flow of control between code segments; an $s$--$t$ path represents a complete entry-to-exit test execution of the program, and forbidden pairs encode mutually unexecutable branches. In this setting, \textsc{PAFP} was proposed as a means to make software testing more efficient by only considering paths through a program that contain at most one branch of each pair of mutually unexecutable branches~\cite{krause73}. 

\textsc{PAFP} also has applications in bioinformatics, in particular in gene finding with RT-PCR evidence. Here, candidate transcripts are modeled by $s$--$t$ paths in a graph in which the vertices are non-overlapping DNA segments and the edges indicate that one segment may immediately follow another in a transcript. By having the forbidden-pair set forbid combinations of DNA segments incompatible with the RT-PCR evidence, \textsc{PAFP} is able to help gene-prediction algorithms rule out biologically implausible transcripts while incorporating RT-PCR data to improve sensitivity and reveal novel splicing variants~\cite{kovac09rtpcr}.

In application settings of \textsc{PAFP}, repeated path computations on large instances typically have to be done, so we would not an exponential-time search here. This difficulty is already present on
directed acyclic graphs: \textsc{PAFP} is \textsf{NP}-complete
even when the input digraph is acyclic~\cite{gabow76}.

\emph{Union/constraint graph.}
A mainstream way to analyze the complexity of constrained graph problems is to encode constraints as edges between vertices,
the \emph{primal graph} viewpoint in constraint satisfaction~\cite{samer_szeider10}.
\textsc{PAFP} fits this unusually well: instances contain adjacency constraints (edges of~$G$) and incompatibility
constraints (pairs in~$\mathcal F$) on the same vertex set. Bodlaender--Jansen--Kratsch exploit this interaction-graph view for forbidden-pairs path
problems, obtaining algorithmic and kernel results under structural restrictions on an undirected union/constraint
graph that contains the (undirected) edges of $G$ together with an edge for each forbidden pair (e.g.\ treewidth-based parameters)~\cite{bodlaender13}. 

For \textsc{PAFP} it is common to orient forbidden pairs consistently with an order on the vertex set (topological or reachability-based) and analyze the ensuing structure~\cite{kolman09,kovac13}. Accordingly, our study of how layer-based width parameters affect \textsc{PAFP} includes consideration of the union graph obtained by adding the forbidden pairs as arcs to the input graph, where each forbidden pair is oriented according to a fixed reachability-compatible order (see Section~\ref{sec:def}). The reachability-compatible order rule $G \mapsto \prec_G$ is fixed only to make the union digraph (and hence its BFS-width) uniquely defined. The particular reachability-compatible order rule is immaterial: all of our results hold for any fixed choice of $G \mapsto \prec_G$ (Remark~\ref{rem:order-choice}).

\emph{BFS-width.}
BFS layerings (grouping vertices by shortest-path distance from a root) are a standard structural tool in graph theory and graph algorithms, underlying, for example, layered path decompositions and separator theorems for minor-closed graph classes~\cite{bannister19,dujmovic17}. In particular, BFS layerings are featured in

recent parameterized algorithms for length-constrained path problems which bound the number of ‘reused’ bfs layers as a function of other parameters~\cite{BezakovaCDF17}. A natural associated width statistic is the maximum size of any BFS distance layer.
This statistic has been studied as a width parameter in the undirected setting, with emphasis on its polynomial-time
computability~\cite{eppstein25}. Since \textsc{PAFP} is posed on directed graphs, in our work ``BFS-width'' and ``BFS layer'' denote the directed versions of these concepts unless stated otherwise, measured in terms of directed shortest-path distances (see Section~\ref{sec:def} for a formal definition).

\textsc{PAFP} is an inherently directed $s$--$t$ path problem, so distance layerings rooted at $s$ are a natural structural lens; in particular any instance, if it has a safe $s$--$t$ path, has a shortest one. Despite this, BFS-width has not been studied for \textsc{PAFP}. We initiate this direction by asking whether bounded BFS-width of the union digraph can help to make \textsc{PAFP} tractable, analogously to how the undirected variant of \textsc{PAFP} is fixed-parameter tractable parameterized by the treewidth of the undirected union/constraint graph~\cite{bodlaender13}.

There is a positive baseline for this question: \textsc{PAFP} itself is indeed fixed-parameter tractable parameterized by the undirected BFS-width of the underlying undirected union/constraint graph\footnote{Let $U$ be the part reachable from $s$ of the underlying undirected union/constraint graph, i.e., the graph obtained  by adding an undirected edge $uv$ for each forbidden pair $\{u,v\}\in\mathcal F$ to the underlying undirected graph of the digraph $G$. If $L_0,L_1,\ldots$ are the undirected BFS layers of $U$ from $s$, then every edge of $U$ has endpoints in the same or adjacent layers, and each vertex appears in at most two consecutive bags $B_d:=L_d\cup L_{d+1}$. Hence $(B_d)_d$ is a path decomposition of width at most $2\,\mathrm{bfsw}_{\mathrm{und}}(U,s)-1$, and therefore $\operatorname{tw}(U)\le 2\,\mathrm{bfsw}_{\mathrm{und}}(U,s)-1$. We then are able to apply the MSO-argument in~\cite[Prop.~6]{bodlaender13} to get fixed-parameter tractability: their argument applies exactly here, with the addition that we need to take the additional step to encode directed edges as binary relation, which can be done in MSO~\cite[Sec.~3, Def.~3.1, pp.~314--315]{arnborgLS91}. Appendix~\ref{app:mso}, which  provides additional details for the proof of Theorem~\ref{thm:bfs-backward-fpt}, explicitly details this encoding.}. This raises the possibility that for the union digraph which stays true to \textsc{PAFP}'s directed nature, tractability can likewise be achieved by bounding its BFS-width.

\subsection{Our results} We show that bounded BFS-width of the union digraph gives tractability
for \textsc{PAFP} when, at every BFS-layer cut, the smaller of the numbers
of distinct heads and distinct tails of crossing backward arcs is bounded
(Theorem~\ref{thm:bfs-backward-fpt}). The smaller endpoint set may be
chosen independently at each cut, and the total number of backward arcs
need not be bounded. The union digraph is helpful as a necessary structural tool here, because inspection of the classical layered reduction of
Gabow et al.\ shows that \textsc{PAFP} is \textsf{NP}-complete on instances for which the input graph has both BFS-width 3 and no backward arcs (Proposition~\ref{prop:input-width3-no-backward-hard}).

In general, we show that not only does bounded BFS-width of the union digraph fail to obtain tractability for \textsc{PAFP}, but that every \textsc{PAFP} instance whose input graph is a DAG has an equivalent polynomial-time computable
\emph{normal form} whose union digraph is a DAG of BFS-width at most~$2$ --- the smallest nontrivial value
(Theorem~\ref{thm:pl-to-bfsw2}). Hence \textsc{PAFP} remains \textsf{NP}-complete even in this ultra-thin,
ladder-like regime where a shortest-path exploration has at most two frontier choices at every distance
(Theorem~\ref{thm:bfsw2-union-npc}). 

We also use the BFS-width 2 normal form to show that \textsc{PAFP} is \textsf{NP}-complete when the input graph has BFS-width 2 (Corollary~\ref{cor:bfsw2-input-hard}). We contrast this with our result that \textsc{PAFP} is tractable when the input graph is a DAG of BFS-width 2 with a bounded number of backward arcs (Theorem~\ref{thm:input-bfsw2-few-backward-fpt}), which again shows backward arcs between BFS layers to be a useful parameter.

To complete our study of layer-based width for \textsc{PAFP}, we also consider a second rooted layering notion on DAGs, namely
exact-length layers $E_G(s,d)$, which record which vertices can occur at the
$d$th position of a directed path from $s$. For this notion, we show that
\textsc{PAFP} is polynomial-time solvable when $\mathrm{elw}(G,s)\le 2$
(Theorem~\ref{thm:exactwidth2}). This width bound is tight, as analysis of the reduction of
Gabow et al.\ shows that \textsc{PAFP} remains \textsf{NP}-complete already on DAG
instances of exact-length width~$3$ (Proposition~\ref{prop:elw3-hard}). 

While there is substantial work on tractable special cases of \textsc{PAFP}, polynomial-time cases thereof rely on additional restrictions on the forbidden-pair set (see Table~\ref{tab:tractable-comparison}). By contrast, our Theorems~\ref{thm:input-bfsw2-few-backward-fpt} and ~\ref{thm:exactwidth2} simultaneously allow (i) unrestricted forbidden pairs and (ii) input graphs with exponentially many $s$--$t$ paths (Remark~\ref{rem:elw-exp}).
\setlength{\textfloatsep}{7pt}
\setlength{\intextsep}{20pt}

\newcommand{\reftop}[3][0.9ex]{%
  \raisebox{#1}[0pt][0pt]{%
    \parbox[t]{#2}{\RaggedRight #3}%
  }%
}

\renewcommand{\tabularxcolumn}[1]{m{#1}}
\newcolumntype{Y}{>{\RaggedRight\arraybackslash}X}

\newcolumntype{L}[1]{>{\RaggedRight\arraybackslash}m{#1}}
\newcolumntype{C}[1]{>{\Centering\arraybackslash}m{#1}}
\begin{table}[!htbp]
\centering
\footnotesize
\setlength{\tabcolsep}{4.5pt}
\renewcommand{\arraystretch}{1.525}
\setlength{\arrayrulewidth}{0.6pt}
\begin{threeparttable}
\caption{Survey of tractable regimes for \textsc{PAFP}.
Known polynomial-time results restrict the forbidden-pair set, either by directly imposing structure on $\mathcal{F}$ or by limiting how $\mathcal{F}$ can interact with $G$. By contrast,
Theorems~\ref{thm:input-bfsw2-few-backward-fpt} and \ref{thm:exactwidth2} allow for arbitrary, unrestricted
$\mathcal{F}$ while still allowing for exponentially many $s$--$t$ paths in the input graph.}
\label{tab:tractable-comparison}
\scalebox{1}{%
\begin{minipage}{\linewidth}
\begin{tabularx}{\linewidth}{|L{1.8cm}|Y|C{2.15cm}|}
\hline
\textbf{Reference} & \textbf{Restriction} & \textbf{Unrestricted $\mathcal{F}$?} \\
\hline

\textbf{This paper}
& $\mathrm{bfsw}(G,s)\le 2$, $G$ is a DAG, and backward arcs in $G$ are bounded; 
$\mathcal{F}$ can be any subset of $\binom{V}{2}$. 
& \textbf{\cmark} \\
\hline

\textbf{This paper}
& $\mathrm{elw}(G,s)\le 2$ and $G$ is a DAG; 
$\mathcal{F}$ can be any subset of $\binom{V}{2}$. 
& \textbf{\cmark} \\
\hline

Kov\'a\v{c} (2013)~\cite{kovac13}
&  For a fixed topological order $\prec$ of $G$, any two distinct forbidden pairs $\{u,v\}$ and $\{x,y\}$ with $u \prec v$, $x \prec y$, and $u \prec x$ satisfy
$x \prec v \prec y$.

& \xmark \\
\hline

Kolman--Pangr\'ac (2009)~\cite{kolman09}
& \emph{Hierarchical forbidden pairs}: With respect to the reachability order $\prec$ of $G$,
no two oriented forbidden pairs $(u,v)$ and $(x,y)$ satisfy $u \prec x \prec v \prec y$.
& \xmark \\
\hline

Bodlaender--Jansen--Kratsch (2013)~\cite{bodlaender13}\tnote{a}
& Bounded treewidth of the underlying undirected union/constraint graph. This restricts $\mathcal{F}$ from making the union/constraint graph have large treewidth.
& \xmark \\
\hline

Yinnone (1997)~\cite{Yinnone1997}
& \emph{Skew-symmetry}: For every $\{u,u'\}, \{v,v'\}\in \mathcal{F}$, the arc condition
$(u,v)\in E \Rightarrow (v',u')\in E$ holds.
& \xmark \\
\hline
\end{tabularx}
\end{minipage}%
}
\vspace{4pt}
\begin{tablenotes}[flushleft]
\footnotesize
\item[a] This result was formulated in an undirected setting for \textsc{PAFP}; for the exact directed setting in this paper, combine
\cite[Prop.~6]{bodlaender13} with bounded-treewidth MSO model
checking for structures with directed edge relations
\cite[Sec.~3, Def.~3.1, pp.~314--315]{arnborgLS91}. Appendix~\ref{app:mso} explicitly spells out the formal details of this combination.

\end{tablenotes}

\end{threeparttable}
\end{table}

\section{Definitions}\label{sec:def}

\paragraph{Graph conventions and \textsc{PAFP}.}
All digraphs in this paper are finite and loopless. For a digraph
$G=(V,E)$, an arc $(u,v)\in E$ has tail $u$ and head $v$.
A directed path is vertex-simple unless explicitly stated otherwise.

A \textsc{PAFP} instance is a tuple
$
I=(G=(V,E),s,t,\mathcal F)
$,
where $s,t\in V$, $s\neq t$, and
$\mathcal F\subseteq \binom V2$. A directed $s$-$t$ path $P$
is safe if it contains at most one endpoint of every pair in
$\mathcal F$. The \textsc{PAFP} problem asks whether such a safe
$s$-$t$ path exists.

\paragraph{Layer widths.} Let $G = (V,E)$ be a directed graph. For $u,v\in V$, let $\operatorname{dist}_G(u,v)$ be the length of a
shortest directed $u$-$v$ path, with $\operatorname{dist}_G(u,v)=\infty$
if no such path exists.
For $d\in\mathbb Z_{\ge 0}$, define the directed BFS layer from $s$ by
\[
L_G(s,d)\coloneq \{v\in V : \operatorname{dist}_G(s,v)=d\},
\]
and define the BFS-width from $s$ by
$
\operatorname{bfsw}(G,s)\coloneq
\max_{d\in\mathbb Z_{\ge 0}} |L_G(s,d)|
$.
Define the exact-length layer
\[
E_G(s,d)\coloneq
\{v\in V : \exists \text{ a directed $s$-$v$ path of length exactly $d$}\},\]
and define
$
\operatorname{elw}(G,s)\coloneq
\max_{d\in\mathbb Z_{\ge 0}} |E_G(s,d)|
$.

\paragraph{Reachability-compatible orders.} Let $G = (V,E)$ be a directed graph, and 
for vertices $u,v\in V$, write $u\leadsto_G v$ if $u\neq v$
and $G$ contains a directed $u$-$v$ path. A strict total order
$\prec$ on $V$ is reachability-compatible with $G$ if, for all
$u,v\in V$,
$
u\leadsto_G v \land v\not\leadsto_G u
\implies
u\prec v
$.
Such orders always exist, as can be seen from topologically ordering the condensation DAG of $G$
and breaking ties deterministically inside SCCs and
among incomparable components. When $G$ is a DAG, a  reachability-compatible
orders is a topological orders.

\paragraph{$\triangleleft$-oriented forbidden-pair arcs.} Let $G=(V,E)$ be a directed graph, let $\mathcal F \subseteq \binom{V}{2}$, and let
$\triangleleft$ be a strict total order on $V$. Define the $\triangleleft$-oriented forbidden-pair arcs
$
A_{\triangleleft} (\mathcal{F})\coloneq \{(u,v)\mid \{u,v\}\in\mathcal F,\ u \triangleleft v \}$.
\paragraph{Union digraph.}
Fix a function $G \mapsto \prec_G$ that maps each digraph $G = (V,E)$ to a  reachability-compatible strict total order $\prec_G$ on $V$.
For a \textsc{PAFP} instance $I=(G=(V,E),s,t,\mathcal F)$, define the
union digraph
$
G\cup\mathcal F
\coloneq
(V,E\cup A_{\prec_G}(\mathcal F))$.

\paragraph{Backward arcs.}
Let $H=(V,E)$ be a digraph with root $s$. Let
$
R_H(s)\coloneq \{v\in V:\operatorname{dist}_H(s,v)<\infty\}
$
be the set of vertices reachable from $s$. The set of backward arcs in $H$ from $s$ is \[
B_H(s)
\coloneq
\{(u,v)\in E : u,v \in R_H(s)  \ \text{and} \operatorname{dist}_H(s,v) < \operatorname{dist}_H(s,u) \}.\]

\paragraph{Heads and tails across BFS-layer cuts.}
We now define, for a non-empty BFS-layer, the set of backward arcs crossing that layer. For the rooted digraph $H$ above, let
$
D\coloneq\max\{d\in\mathbb Z_{\ge0}:L_H(s,d)\neq\emptyset\}
$.
For each $d\in\{0,\ldots,D-1\}$, define
$
S_d\coloneq\bigcup_{i=0}^{d}L_H(s,i)$ and \[
\mathcal A_d\coloneq
\{(u,v)\in B_H(s):v\in S_d,\ u\in R_H(s)\setminus S_d\}.
\]
Thus $\mathcal A_d$ is the set of backward arcs crossing the cut after
BFS layer $d$. Define their sets of distinct heads and tails as
$
\mathcal H_d\coloneq\{v:\exists u,\ (u,v)\in\mathcal A_d\}$ and $
\mathcal T_d\coloneq\{u:\exists v,\ (u,v)\in\mathcal A_d\}.
$

\begin{remark}\label{rem:order-choice}
The fixed rule $G\mapsto \prec_G$ is used only to make $G\cup \mathcal F$
uniquely defined on arbitrary \textsc{PAFP} instances. Our results hold for any fixed  reachability-compatible order rule. 
\end{remark}

\section{BFS Layers of the Union Digraph}

We commence our study of BFS Layers of the union digraph by showing that
bounded BFS-width, together with a bound on the smaller of the two
backward-arc endpoint counts at each BFS-layer cut, gives tractability. 
\label{subsec:bfs-backward-fpt}

\begin{theorem}[FPT for bounded union BFS-width and backward-cut endpoint counts]
\label{thm:bfs-backward-fpt}
Let
$
I=(G=(V,E),s,t,\mathcal F)
$
be a \textsc{PAFP} instance where $G$ is an arbitrary directed graph, and let
$
H:=G\cup\mathcal F
$
be its union digraph. Set
$
b:=\operatorname{bfsw}(H,s)$ and 
$
p:=\max_{0\le d<D}\min\{|\mathcal H_d|,|\mathcal T_d|\},
$
where $D$, $\mathcal H_d$, and $\mathcal T_d$ are defined for $H$
and root $s$ as above, and $p=0$ when $D=0$.

Then $I$ can be decided in time
$
f(b+p)\cdot |I|^{O(1)}$
for some computable function $f$. 
\end{theorem}

\begin{proof}
As a tool in the analysis, this proof talks about the concept of the auxiliary undirected union/constraint graph; our parameters $b$ and $p$ here are undefined for it. For a vertex set $S\subseteq V$, define the auxiliary undirected
union/constraint graph restricted to $S$ by
$
U[S]\coloneq
\bigl(S,\,
\{\{u,v\}\in \binom S2 :
(u,v)\in E \text{ or } (v,u)\in E \text{ or } \{u,v\}\in\mathcal F\}
\bigr).
$

Let
$
R_H:=\{v\in V:\operatorname{dist}_H(s,v)<\infty\}$
and
$
H_R:=H[R_H]$.
We first show that $U[R_H]$ has pathwidth at most
$
2b+2p-1$. For $d\ge 0$, write
$
L_d:=L_{H_R}(s,d),
$
and let
$
D:=\max\{d:L_d\neq\emptyset\}$.
All distances below are taken in $H_R$. These distances agree with the
corresponding distances in $H$: if $v\in R_H$, then every vertex on a
shortest $s$-to-$v$ path in $H$ is itself reachable from $s$ in $H$, and
hence lies in $R_H$.
Let $Q$ be the undirected graph on $R_H$ with edge set
$
E(Q):=\{\{u,v\}:(u,v)\in B_H(s)\}.
$
For $d=0,\ldots,D-1$, let $Q_d$ be the bipartite graph consisting of
all edges of $Q$ between $S_d$ and $R_H\setminus S_d$, where $S_d$
is as in the definitions. Its edges are precisely
$\{\{u,v\}:(u,v)\in\mathcal A_d\}$.
Choose the canonical $S_d$-K\"onig minimum vertex cover $C_d$ of
$Q_d$, as defined in \cite[Section~3.1]{JeongST15}.
Explicitly, take a maximum matching of $Q_d$, and let $W_d$ be the
vertices reachable from unmatched vertices of $S_d$ by alternating paths
that traverse nonmatching edges from $S_d$ to its complement and matching
edges in the reverse direction. Then
$
C_d=(S_d\setminus W_d)\cup\bigl((R_H\setminus S_d)\cap W_d\bigr).
$
Both $\mathcal H_d$ and $\mathcal T_d$ are vertex covers of $Q_d$.
Since $C_d$ is a minimum vertex cover, we have that 
$
|C_d|\le\min\{|\mathcal H_d|,|\mathcal T_d|\}\le p.
$
Now set $C_{-1}=C_D=\emptyset$ and $L_{D+1}=\emptyset$, and define, for $d \in \{0 , \ldots, D\}$, 
$
X_d:=L_d\cup L_{d+1}\cup C_{d-1}\cup C_d
.
$
We claim that $(X_0,X_1,\ldots,X_D)$ is a path decomposition of $U[R_H]$, which we argue next.

First, every vertex of $R_H$ is covered: if $v\in L_i$, then
$v\in X_i$, and if $i\ge1$, also $v\in X_{i-1}$.
For the interval condition of a path decomposition, which here is that for all $v\in R_H$ and all
    $0\le i<j<k\le D$,
    $
    v\in X_i\cap X_k\implies v\in X_j.
    $, applying \cite[Lemma~3]{JeongST15} to the fixed
graph $Q$ and the tripartition
$(S_d,R_H\setminus S_e,S_e\setminus S_d)$ gives, for $0\le d<e<D$, that
$
C_e\cap S_d\subseteq C_d$ and $C_d\cap(R_H\setminus S_e)\subseteq C_e.
$
Consequently, for $v\in L_i$, its occurrences among
$C_0,\ldots,C_{i-1}$ form a suffix, possibly empty, and its occurrences
among $C_i,\ldots,C_{D-1}$ form a prefix, possibly empty.
Each $C_d$ occurs in both $X_d$ and $X_{d+1}$. Together with the
relevant occurrences in $X_i$ and, when $i>0$, $X_{i-1}$, we have that
that the bags containing $v$ form a contiguous interval.

It remains to verify the edge coverage condition of a path decomposition, which here is that for every edge $\{u,v\}$ of $U[R_H]$,
    there exists $d\in\{0,\ldots,D\}$ such that
    $
    \{u,v\}\subseteq X_d.
    $. Let $\{x,y\}$ be an edge of $U[R_H]$.
By definition of $U[R_H]$, either one of $(x,y)$ and $(y,x)$ is an input
arc of $G$, or $\{x,y\}\in\mathcal F$. In either case, the union digraph
$H$ contains at least one directed arc whose underlying undirected edge is
$\{x,y\}$. Since $x,y\in R_H$, this arc lies in $H_R$. Write such an arc as
$(a,c)$, where $\{a,c\}=\{x,y\}$.
If $(a,c)\notin B_H(s)$, then
$
\operatorname{dist}_{H_R}(s,a)
\le\operatorname{dist}_{H_R}(s,c)
\le\operatorname{dist}_{H_R}(s,a)+1.
$
Thus its endpoints lie in the same or consecutive BFS layers and are
contained together in some $X_d$.

Otherwise, let $c\in L_i$ and $a\in L_j$ for $i<j$.
$\{a,c\}$ belongs to $Q_d$ for each $d=i,\ldots,j-1$, so
each corresponding cover $C_d$ contains $a$ or $c$.
If $a\in C_i$, then $X_i$ contains both endpoints.
Otherwise $c\in C_i$. Let $r$ be the largest index in
$\{i,\ldots,j-1\}$ such that $c\in C_r$.
If $r=j-1$, then $X_{j-1}$ contains both $c\in C_{j-1}$ and
$a\in L_j$. If $r<j-1$, then $c\notin C_{r+1}$, so
$a\in C_{r+1}$, and $X_{r+1}$ contains both endpoints via
$C_r\cup C_{r+1}$.
Hence every edge of $U[R_H]$ is covered by some bag, and
$(X_d)_{d=0}^D$ is a path decomposition. Finally, for every $d$ we have that 
$
|X_d|\le |L_d|+|L_{d+1}|+|C_{d-1}|+|C_d|\le 2b+2p,
$ so
$
\operatorname{pw}(U[R_H])\le\max_d|X_d|-1\le 2b+2p-1.
$

Now let
$
R_G:=\{v\in V:\operatorname{dist}_G(s,v)<\infty\}$.
If $t\notin R_G$, then no directed $s$-$t$ path exists in $G$, and the
instance is trivially a NO-instance. Hence assume $t\in R_G$.
Let
$
\mathcal F_G:=\{\{u,v\}\in\mathcal F:u,v\in R_G\},
$
and consider the restricted instance
$
I_G:=(G[R_G],s,t,\mathcal F_G)$.
The instances $I$ and $I_G$ are equivalent, since every vertex appearing on
a directed $s$-$t$ path in $G$ is reachable from $s$ in $G$, and hence
lies in $R_G$.

Because $G$ is a subdigraph of $H$, every vertex reachable from $s$ in
$G$ is also reachable from $s$ in $H$. Thus
$
R_G\subseteq R_H$.
Consequently $U[R_G]$ is an induced subgraph of $U[R_H]$. Since
pathwidth is monotone under taking induced subgraphs, we have that
$
\operatorname{pw}(U[R_G])\le 2b+2p-1$.

Now that  \textsc{PAFP} is FPT w.r.t. the treewidth of the undirected version of the union graph gives us the desired result here. The referenced statement itself follows from~\cite[Prop.~6]{bodlaender13} and~~\cite[Sec.~3, Def.~3.1, pp.~314--315]{arnborgLS91}; Appendix~\ref{app:mso}  clarifies how these sources actually show the treewidth result

Since
$
\operatorname{tw}(U[R_G])
\le
\operatorname{pw}(U[R_G])
\le
2b+2p-1$,
the restricted instance $I_G$, and therefore also the original instance $I$,
can be decided in time
$
f(b+p)\cdot |I|^{O(1)}
$
for some computable function $f$.\qedbox
\end{proof}

We now develop the normal form which allows us to show that \textsc{PAFP} remains \textsf{NP}-complete even when the BFS-width of the input graph of its instances is restricted to be the smallest non-trivial value. 

\subsection{BFS-width-2 normal form}
\subsubsection{Intuition and overview}
The intuition behind our normal form is that we turn an arbitrary DAG \textsc{PAFP} instance into an ``ultra-thin'' ladder from a fresh start vertex $s'$: we build a long directed path of new placeholder vertices (the \emph{spine}) and then give each original reachable vertex a short detour (i.e., a directed path of length two via a fresh intermediate vertex) off a designated spine position; see Figure~\ref{fig:normalize-optionA} for an illustration. Thus each original vertex is forced to appear at a controlled shortest-path distance from $s'$. There is a careful placement of the detours and in addition the inclusion of their intermediate vertices, so as to prevent  prevent multiple non-spine vertices from landing in the same BFS layer to keep the BFS layers narrow.

\noindent\emph{Distance subtlety}
The main subtlety for the BFS-width bound is preventing the \emph{core union arcs} from creating unintended \emph{distance-shortcuts}: once we consider the union digraph (i.e., the original arcs of $G$ together with the $\triangleleft$-oriented forbidden-pair arcs on the original vertices), either type of core arc could in principle combine with the spine and detours to reach some vertex earlier than its intended detour distance. We avoid this by using a reachability-compatible order $\triangleleft$ of the reachable core as a \emph{scheduling constraint} for the detours. Concretely, we assign vertices to detour positions in reverse-$\triangleleft$ order (so vertices later in $\triangleleft$ are attached closer to $s'$). 

Think of the spine as a timeline, in which  attaching a vertex closer to $s'$ means it becomes reachable earlier. If a union arc $u\to v$ jumps forward by $a$ positions in $\triangleleft$, then we ``wait'' and place the detour for $u$ \emph{after} the detour for $v$ by at least $a$ detour-slots along the spine. Thus $v$ is scheduled $a$ slots closer to $s'$ than $u$, so traversing $u\to v$ can never beat the direct detour into $v$ and cannot create a new shortest-path shortcut.

Equisatisfiability is enforced by ``booby-trapping'' each detour using its intermediate vertex: we add a forbidden pair that makes taking that detour unsafe, except for the unique detour that enters the original start $s$. Consequently, every safe $s'$--$t$ path is forced to follow the spine into $s$ and then continue inside the original instance. Finally, as a bookkeeping step, we insert the reachable-core forbidden-pair arcs oriented forward in the reachable core's order $\triangleleft$ as ordinary arcs of the output graph $G'$. This ensures that forming the union digraph of the output instance introduces no additional arcs (so we may analyze BFS distances directly in $G'$), and it aligns all core arcs with the same order used in the detour scheduling. This step does not change satisfiability, since any path traversing such an inserted arc would contain both endpoints of that forbidden pair and hence be unsafe. 


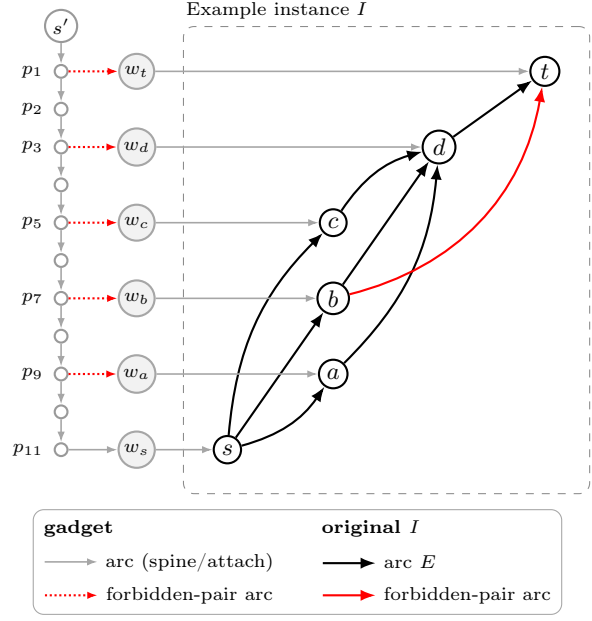
\begin{figure}[t] 
\centering

\tikzset{
  e/.style={-{Latex[length=2mm]}, thick, draw=black},
  h/.style={-{Latex[length=2mm]}, thick, draw = red},
  g/.style={-{Latex[length=1.4mm,width=1.0mm]}, line width=0.6pt, draw=gray!70},
  f/.style={-{Latex[length=1.4mm,width=1.0mm]}, line width=0.8pt, draw=red, densely dotted},
  normalize/.style={
    >=Latex,
    origV/.style={circle, draw=black, thick, inner sep=1.6pt, font=\small},
    gadV/.style={circle, draw=gray!80, thick, fill=white, minimum size=4.8pt, inner sep=1.2pt, font=\scriptsize},
    detV/.style={circle, draw=gray!70, thick, fill=gray!10, inner sep=1.2pt, font=\scriptsize},
    lab/.style={font=\scriptsize, fill=white, inner sep=1pt}
  }
}

\begin{minipage}{\linewidth}
\centering

\begin{tikzpicture}[normalize]

\node[origV] (s) at (0,-3) {$s$};
\node[origV] (a) at (1.4,-2) {$a$};
\node[origV] (b) at (1.4,-1) {$b$};
\node[origV] (c) at (1.4, 0) {$c$};
\node[origV] (d) at (2.8, 1) {$d$};
\node[origV] (t) at (4.2, 2) {$t$};

\draw[e] (s) to[bend right=20] (a);
\draw[e] (s) -- (b);
\draw[e] (s) to[bend left = 20](c);
\draw[e] (a) to[bend right=20] (d);
\draw[e] (b) -- (d);
\draw[e] (c) to[bend left=20] (d);
\draw[e] (d) -- (t);

\draw[h] (b) to[bend right = 33] (t);

\def\spx{-2.2}
\node[gadV] (sp) at (\spx,2.6) {$s'$};

\foreach \i in {1,...,11} {
  \pgfmathsetmacro{\yy}{2 - 0.5*(\i-1)}
  \node[gadV] (p\i) at (\spx,\yy) {};
}

\begin{scope}[local bounding box=plabelbox]
  \foreach \i in {1,2,3,5,7,9,11} {
    \node[font=\scriptsize, anchor=east] at ($(p\i)+(-0.12,0)$) {$p_{\i}$};
  }
\end{scope}

\draw[g] (sp) -- (p1);
\foreach \j in {1,...,10} {
  \pgfmathtruncatemacro{\k}{\j+1}
  \draw[g] (p\j) -- (p\k);
}

\node[detV] (wt) at (-1.2, 2) {$w_t$};
\node[detV] (wd) at (-1.2, 1) {$w_d$};
\node[detV] (wc) at (-1.2, 0) {$w_c$};
\node[detV] (wb) at (-1.2,-1) {$w_b$};
\node[detV] (wa) at (-1.2,-2) {$w_a$};
\node[detV] (ws) at (-1.2,-3) {$w_s$};

\draw[g] (wt) -- (t);
\draw[g] (wd) -- (d);
\draw[g] (wc) -- (c);
\draw[g] (wb) -- (b);
\draw[g] (wa) -- (a);
\draw[g] (ws) -- (s);

\draw[f] (p1) -- (wt);
\draw[f] (p3) -- (wd);
\draw[f] (p5) -- (wc);
\draw[f] (p7) -- (wb);
\draw[f] (p9) -- (wa);

\draw[g] (p11) -- (ws);

\begin{scope}[on background layer]
  \node[draw=black!45, dashed, rounded corners, inner sep=11pt,
        fit=(s)(a)(b)(c)(d)(t)] (corebox) {};
\end{scope}

\node[lab, anchor=south west] at ($(corebox.north west)+(0,0.06)$)
  {Example instance $I$};

\end{tikzpicture}

\vspace{0.5em}

\begin{tikzpicture}[normalize, font=\scriptsize]
\matrix[
  draw=black!35,
  fill=white,
  rounded corners,
  inner sep=2pt,
  row sep=2pt,
  column sep=14pt,
  nodes={anchor=west, align=left}
] {
  \node {\textbf{gadget}}; & \node {\textbf{original $I$}}; \\
  \node {\tikz[normalize,baseline=-0.6ex]{\draw[g] (0,0) -- (0.70,0);}~arc (spine/attach)}; &
  \node {\tikz[normalize,baseline=-0.6ex]{\draw[e] (0,0) -- (0.70,0);}~arc $E$}; \\
  \node {\tikz[normalize,baseline=-0.6ex]{\draw[f] (0,0) -- (0.70,0);}~forbidden-pair arc}; &
  \node {\tikz[normalize,baseline=-0.6ex]{\draw[h] (0,0) -- (0.70,0);}~forbidden-pair arc}; \\
};
\end{tikzpicture}

\caption{The union digraph $G'\cup\mathcal F'$ of the normalized instance
$I'=\operatorname{Normalize}(I)$, taking $I$ to be the boxed structure.
The reachable-core forbidden-pair arc $b\to t$ is inserted into the output
graph according to the internal order $\triangleleft$. Since the detour into
$t$ leaves the spine at $p_1$, while the detour into $b$ leaves later at
$p_7$, traversing $b\to t$ cannot create a shorter $s'$-to-$t$ route
than the direct detour into $t$.}

\label{fig:normalize-optionA}

\end{minipage}
\end{figure}

\begin{theorem}\label{thm:pl-to-bfsw2}
There exists a polynomial-time computable mapping $\operatorname{Normalize}$ such that for every \textsc{PAFP} instance
$I=(G,s,t,\mathcal F)$ where $G$ is a DAG, the \textsc{PAFP} instance $I' := \operatorname{Normalize}(I) = (G',s',t,\mathcal F')$
satisfies that $G'$ is a DAG and $\mathrm{bfsw}(G' \cup \mathcal F',s') \le 2$, and $I \text{ is a \textsc{YES}-instance} \iff I' \text{ is a \textsc{YES}-instance}$.
\end{theorem}
\begin{definition}[Normalization mapping $\operatorname{Normalize}$]\label{def:normalize}
Fix a \textsc{PAFP} instance $I=(G=(V,E),s,t,\mathcal F)$ where $G$ is a DAG.
Define $\operatorname{Normalize}(I)=(G'=(V',E'),\,s',\,t,\,\mathcal F')$ as follows.

\noindent\textbf{Reachable core.}
Let
\begin{equation} V_{\mathrm{rch}} \;\coloneq\; \{v\in V \mid \mathrm{dist}_G(s,v)<\infty\},
\qquad
E_{\mathrm{rch}} \;\coloneq\; E\cap (V_{\mathrm{rch}}\times V_{\mathrm{rch}}), \end{equation}
and
\begin{equation}
\mathcal F_{\mathrm{rch}} \;\coloneq\; \bigl\{\{u,v\}\in \mathcal F \mid u,v\in V_{\mathrm{rch}}\bigr\}.
\end{equation}
Let $G_{\mathrm{rch}}\coloneq (V_{\mathrm{rch}},E_{\mathrm{rch}})$. Fix a deterministic polynomial-time linear-extension routine
$\mathsf{LE}$ for finite DAGs: on a DAG $D=(W,A)$, the routine
returns a strict total order $\mathsf{LE}(D)$ of $W$ such that
$u\leadsto_D v \implies u\,\mathsf{LE}(D)\,v$. Let $\triangleleft \coloneq \mathsf{LE}(G_{\mathrm{rch}})$. Set $E_0\coloneq E_{\mathrm{rch}}\cup A_{\triangleleft} (\mathcal F_{\mathrm{rch}})$, where  $A_\triangleleft(\mathcal F_{\mathrm{rch}})$ is as defined in Section~\ref{sec:def}.

\noindent\textbf{Reverse-$\triangleleft$ enumeration.}
Let $q\coloneq |V_{\mathrm{rch}}|$ and fix an ordering $(r_1,\ldots,r_q)$ of $V_{\mathrm{rch}}$ such that $r_q\triangleleft r_{q-1}\triangleleft\cdots\triangleleft r_1$.
Note that $q \ge 1$ since $s\in V_{\mathrm{rch}}$.
If $q=1$, then $V_{\mathrm{rch}}=\{s\}$ and (since $s\neq t$) we have $t\notin V_{\mathrm{rch}}$,
so the instance is trivially \textsc{NO}; in this case, index ranges such as $j=1,\ldots,2q-2$ are empty.

\noindent\textbf{gadget vertices.}
Introduce pairwise distinct vertices
\[
s'\notin V,\qquad p_1,\ldots,p_{2q-1}\notin V\cup\{s'\},\qquad
w_1,\ldots,w_q\notin V\cup\{s'\}\cup\{p_1,\ldots,p_{2q-1}\}.
\]
Let $P\coloneq \{p_1,\ldots,p_{2q-1}\}$ and $W\coloneq \{w_1,\ldots,w_q\}$.

\noindent\textbf{Arcs.}
Define
\begin{equation}
E_{\mathrm{spine}} \;\coloneq\; \{(s',p_1)\}\ \cup\ \{(p_j,p_{j+1}) \mid j=1,\ldots,2q-2\},
\end{equation}
\begin{equation}
E_{\mathrm{att}} \;\coloneq\; \{(p_{2i-1},w_i)\mid i=1,\ldots,q\}\ \cup\ \{(w_i,r_i)\mid i=1,\ldots,q\},
\end{equation}
and set
\begin{equation}
V' \;\coloneq\; V_{\mathrm{rch}} \cup \{t\} \cup \{s'\} \cup P \cup W,
\qquad
E' \;\coloneq\; E_0 \cup E_{\mathrm{spine}} \cup E_{\mathrm{att}}.
\end{equation}

\noindent\textbf{Forbidden pairs.}
Define
\begin{equation}
\mathcal F' \;\coloneq\; \mathcal F_{\mathrm{rch}}
\ \cup\ 
\bigl\{\{p_{2i-1},w_i\} \ \bigm|\ i\in\{1,\ldots,q\}\ \text{and}\ r_i\neq s \bigr\}.
\end{equation}

\noindent Finally, let $G'\coloneq (V',E')$ and output $\operatorname{Normalize}(I)=(G',s',t,\mathcal F')$.
\end{definition}
 \begin{proof}[Proof of Theorem~\ref{thm:pl-to-bfsw2}]

\noindent\emph{Standing notation for $\operatorname{Normalize}$.}
Fix a \textsc{PAFP} instance $I=(G,s,t,\mathcal F)$ where $G = (V,E)$ is a DAG, and let
$I'=(G'=(V',E'),s',t,\mathcal F')\coloneq \operatorname{Normalize}(I)$. Write $V_{\mathrm{rch}},E_{\mathrm{rch}},\mathcal F_{\mathrm{rch}},G_{\mathrm{rch}}$, and $\triangleleft$ for the reachable-core objects from Definition~\ref{def:normalize},
where $E_0 \coloneq E_{\mathrm{rch}}\cup A_{\triangleleft} ( \mathcal{F}_{\mathrm{rch}})$. Let
$(r_1,\ldots,r_q)$ be the fixed reverse-$\triangleleft$ listing of $V_{\mathrm{rch}}$ (with $q\coloneq |V_{\mathrm{rch}}|$),
and let $P=\{p_1,\ldots,p_{2q-1}\}$ and $W=\{w_1,\ldots,w_q\}$ denote the gadget vertices of $G'$.
We also use $E_{\mathrm{spine}}$ and $E_{\mathrm{att}}$ for the spine and attachment arc sets, so that
$E'=E_0\cup E_{\mathrm{spine}}\cup E_{\mathrm{att}}$.

\emph{Output is a valid \textsc{PAFP} instance.}
By Definition~\ref{def:normalize}, the output $I'=(G'=(V',E'),s',t,\mathcal F')$ satisfies
$s',t\in V'$ with $s'\neq t$ (since $s'\notin V$ while $t\in V$),
$E'\subseteq V'\times V'$, and $\mathcal F'\subseteq \binom{V'}{2}$.
Hence $I'$ is a valid \textsc{PAFP} instance.

\emph{Polynomial-time computability.}
The mapping of Definition~\ref{def:normalize} is computable in polynomial time.
First of all, $V_{\mathrm{rch}}$ can be obtained by a single reachability search from $s$ in $G$ (e.g. by DFS),
after which $E_{\mathrm{rch}}$ and $\mathcal F_{\mathrm{rch}}$ can be computed by one scan of $E$ and $\mathcal F$. The next task of obtaining $\triangleleft$ can, by definition, be effectuated in polynomial time. We then orient each pair in $\mathcal F_{\mathrm{rch}}$ according to $\triangleleft$ to obtain
$A_{\triangleleft} ( \mathcal{F}_{\mathrm{rch}})$ and $E_0=E_{\mathrm{rch}}\cup A_{\triangleleft} (\mathcal F_{\mathrm{rch}})$.
Finally, we create $O(|V_{\mathrm{rch}}|)$ fresh vertices and explicitly list the sets
$E_{\mathrm{spine}},E_{\mathrm{att}}$, and the additional $O(|V_{\mathrm{rch}}|)$ ``booby-trap'' forbidden pairs,
and assemble $(V',E',\mathcal F')$ by set unions.
All of these operations run in time polynomial in $|V|+|E|+|\mathcal F|$. 
\begin{claim}[Acyclicity and union-closure of the output]\label{clm:normalize-dag-union}
 $G'=(V',E')$ is a DAG. Moreover,
$
A_{\prec_{G'}}(\mathcal F')\subseteq E'$.
Consequently,
$
G'\cup\mathcal F'=(V',E')$.
\end{claim}

\begin{proof}

To show that $G'$ is a DAG, we will give a topological ordering of it. We know that $G_{\mathrm{rch}}$ is a DAG, because $G$ is a DAG. Also, by construction, for every  $(u,v)\in E_{\mathrm{rch}}$ we have
$
u\triangleleft v$. It follows that
$\triangleleft$ is a topological order of
$
(V_{\mathrm{rch}},E_0) $ using also how 
every arc in
$
A_{\triangleleft}(\mathcal F_{\mathrm{rch}})$
is oriented forward with respect to $\triangleleft$ by definition, for which recall that $E_0=E_{\mathrm{rch}}\cup A_{\triangleleft}(\mathcal F_{\mathrm{rch}})$.

Define the sequence
$
\tau_0
\coloneq
s',p_1,p_2,\ldots,p_{2q-1},
w_1,w_2,\ldots,w_q,
r_q,r_{q-1},\ldots,r_1$; we  define $\tau$ to be $\tau_0$ with $t$ at the end if $t\notin V_{\mathrm{rch}}$, and $\tau$ to be $\tau_0$ itself if otherwise $t\in V_{\mathrm{rch}}$. 

We claim that $\tau$ is a topological ordering of $G'$. To do this, we need to check that every arc in $E'$ features its tail being ordered before its head in $\tau$. We now do this check. As
$
E'=E_0\cup E_{\mathrm{spine}}\cup E_{\mathrm{att}}$, we consider the arcs of $E'$ in the three separate cases of whether they belong to $E_0$, $E_{\mathrm{spine}}$, or $E_{\mathrm{att}}$.

If $(x,y)\in E_{\mathrm{spine}}$, then either $(x,y)=(s',p_1)$, or
$(x,y)=(p_j,p_{j+1})$ for some $j\in\{1,\ldots,2q-2\}$. In both cases,
$\tau$ lists $x$ before $y$.

If $(x,y)\in E_{\mathrm{att}}$, then it must be either
$
(x,y)=(p_{2i-1},w_i) $
or
$
(x,y)=(w_i,r_i) $
for some $i\in\{1,\ldots,q\}$. $\tau$ orders each vertex in the spine before each vertex of the form $w_i$, and it also orders each vertex of the form $w_i$ before each vertex in  $V_{\mathrm{rch}}$. So in this case, $\tau$ orders the arc in the desired way.

If $(x,y)\in E_0$, then also $x,y\in V_{\mathrm{rch}}$, so we have $x\triangleleft y$. By construction --- in particular, the way the $r_i$ are ordered --- the sequence $
r_q,r_{q-1},\ldots,r_1 $ lists its elements in increasing order of how they compare using $\triangleleft$. Therefore, $\tau$ orders $x$ before $y$.

Thus every arc of $E'$ goes forward in $\tau$. Therefore $\tau$ is a
topological ordering of $G'$, and $G'$ is a DAG.

\smallskip

\noindent \emph{Union-closure.}
We now show that forming the union digraph of the output instance adds no new
arcs. That is, we prove that
$
A_{\prec_{G'}}(\mathcal F')\subseteq E'$.
Let $\{x,y\}\in\mathcal F'$. By the definition of $\mathcal F'$, there are
two cases.

\noindent
\emph{Case 1: $\{x,y\}\in\mathcal F_{\mathrm{rch}}$.}
In the construction of $A_{\triangleleft}(\mathcal F_{\mathrm{rch}})$, exactly
one of $(x,y)$ and $(y,x)$ belongs to
$
A_{\triangleleft}(\mathcal F_{\mathrm{rch}})\subseteq E_0\subseteq E'$.
Call this arc $(a,b)$. Since $(a,b)\in E'$, we have
$
a\leadsto_{G'} b$.
Since $G'$ is acyclic, we also have
$
b\not\leadsto_{G'}a$.
 $\prec_{G'}$ is reachability-compatible with respect to $G'$, giving us
$
a\prec_{G'} b$. As a result, $(a,b)$ is in the union from the contribution of $\{x,y\}$.

The other case is that $\{x,y\}$ consists of $p_{2i-1}$ and $w_i$, for some $i\in\{1,\ldots,q\}$ for which $r_i\neq s$. From $
(p_{2i-1},w_i)\in E_{\mathrm{att}}\subseteq E'$ it follows that  $
p_{2i-1}\leadsto_{G'} w_i$. Because $G'$ is acyclic,
$
w_i\not\leadsto_{G'} p_{2i-1}$, and because of  $\prec_{G'}$'s status as being reachability-compatible, we get that
$
p_{2i-1}\prec_{G'} w_i$.

In both cases, the arc we add from the forbidden pair is in $E'$ anyway, so taking the union doesn't change the graph-object; that is, 
$
G'\cup\mathcal F'
=
(V',E'\cup A_{\prec_{G'}}(\mathcal F'))
=
(V',E')$.
\qedtri
\end{proof}

\begin{claim}[BFS-width-$2$ via a level function]\label{clm:normalize-bfsw2-level} If 
$
H\coloneq G'\cup \mathcal F'
$
is the union digraph of the instance $\operatorname{Normalize}$ outputs, then
$
\operatorname{bfsw}(H,s')\le 2$.

\end{claim}

\begin{proof}
$
H=(V',E')$ is but from Claim~\ref{clm:normalize-dag-union}.

Define $\lambda:V'\to \mathbb{Z}_{\ge 0} \cup \{\infty\}$ with
\begin{equation}
\begin{aligned}
&\lambda(s')=0,\quad \lambda(p_j)=j\ (1\le j\le 2q-1),\\
&\lambda(w_i)=2i,\quad \lambda(r_i)=2i+1\ (1\le i\le q).
\end{aligned}
\end{equation}, and take $
\lambda(t)=\infty$ when  $t\notin V_{\mathrm{rch}}$. $
\operatorname{dist}_H(s',v)=\lambda(v)
$ for every vertex  $v\in V'$.  We will show that that $
\operatorname{dist}_H(s',v)=\lambda(v)
$ for every vertex  $v\in V'$. For $s'$ it si trivial, so let's consider the other types of vertices in $V'$ now. 

The first type of vertices we consider to show the mentioned equality are spine vertices. Let $j\in\{1,\ldots,2q-1\}$.$
\operatorname{dist}_H(s',p_j)\le j$, because$
s'\to p_1\to p_2\to\cdots\to p_j
$ has length $j$. Now the unique vertex from which there is an arc from it to $p_1$ is $s'$, and for $j\ge 2$, the only such inward neighbor thereof is $p_{j-1}$. Thus every path from 
$s'$ to $p_j$ must be nothing but each of the vertices of the spine from $s'$ to $p_j$ visited in the order in which they appear on the spine. Because this holds for every $j$, we have $\operatorname{dist}_H(s',p_j)\ge j$ and 
$\operatorname{dist}_H(s',p_j)=j=\lambda(p_j)$.

Now let's consider a vertex of the form $w_i$ for $i\in\{1,\ldots,q\}$.
First of all, we have $
\operatorname{dist}_H(s',w_i)\le 2i$, because $
s'\to p_1\to\cdots\to p_{2i-1}\to w_i
$ contains exactly $2i$ arcs. The only vertex for which it is involved in arc whose head is $w_i$ is  $p_{2i-1}$: thus any path from $s'$ to $w_i$ must contain $p_{2i-1}$. Thence $
\operatorname{dist}_H(s',w_i)
\ge
\operatorname{dist}_H(s',p_{2i-1})+1
=
(2i-1)+1
=
2i$ holds, since every directed path from $s'$ to $w_i$ must end with the arc $(p_{2i-1},w_i)$. This establishes that
$
\operatorname{dist}_H(s',w_i)=2i=\lambda(w_i)$.

\smallskip
\noindent

Consider now the case of a vertex of the form $r_i$ for  $i\in\{1,\ldots,q\}$. The path
$
s'\to p_1\to\cdots\to p_{2i-1}\to w_i\to r_i
$
has length $2i+1$, so
$
\operatorname{dist}_H(s',r_i)\le 2i+1$.

We now think about the reverse inequality. let $Q$ be any directed  path from $s'$ to $r_i$ (in
$H$), and let $r_k$ be the first vertex of $V_{\mathrm{rch}}$ which occurs in
$Q$. The only arcs of $E'$ with their head in
$V_{\mathrm{rch}}$ while having tail outside of $V_{\mathrm{rch}}$ are the
$
(w_\ell,r_\ell) $, for  $\ell=1,\ldots,q$. Yet $H=(V',E')$ the predecessor of $r_k$ on $Q$ has to be
$w_k$, implying
$
|E(Q)|\ge \operatorname{dist}_H(s',w_k)+1=2k+1$.

Upon $Q$ reaching $r_k$ at first, it cannot visit a gadget vertex since no arc is from $V_{\mathrm{rch}}$ to
$
\{s'\}\cup P\cup W$.
Indeed,
$
E_0\subseteq V_{\mathrm{rch}}\times V_{\mathrm{rch}},
$
while every arc in $E_{\mathrm{spine}}\cup E_{\mathrm{att}}$ has its tail in
$\{s'\}\cup P\cup W$. Hence the suffix of $Q$ from $r_k$ to $r_i$ lies
entirely in the core digraph
$
(V_{\mathrm{rch}},E_0)$.

We now show that this is possible only when $k\ge i$. Recall that
$
(r_1,\ldots,r_q)
$
is the reverse-$\triangleleft$ listing of $V_{\mathrm{rch}}$, so
$
r_q\triangleleft r_{q-1}\triangleleft\cdots\triangleleft r_1$.
Also, $\triangleleft$ is a topological order of the core digraph
$
(V_{\mathrm{rch}},E_0)$.
Indeed, every arc of $E_{\mathrm{rch}}$ goes forward with respect to
$\triangleleft$, since $\triangleleft=\mathsf{LE}(G_{\mathrm{rch}})$ extends
proper reachability in $G_{\mathrm{rch}}$; and every arc of
$A_{\triangleleft}(\mathcal F_{\mathrm{rch}})$ goes forward with respect to
$\triangleleft$ by definition.

Thus, if $(r_a,r_b)\in E_0$, then
$
r_a\triangleleft r_b$.
Because the list $(r_1,\ldots,r_q)$ is in decreasing $\triangleleft$-order,
this implies
$
a>b$.
Consequently, along any directed path
$
r_{i_0}\to r_{i_1}\to\cdots\to r_{i_\ell}
$
in $(V_{\mathrm{rch}},E_0)$ it must be the case that the indices strictly decrease, i.e.
$
i_0>i_1>\cdots>i_\ell$. That is,  suffix of $Q$ can reach $r_i$ from $r_k$ only if
$
k\ge i$.
It follows that
$
|E(Q)|\ge 2k+1\ge 2i+1=\lambda(r_i)$.
Since $Q$ was arbitrary,
$
\operatorname{dist}_H(s',r_i)\ge 2i+1$.
Together with the straightforward detour path, this gives
$
\operatorname{dist}_H(s',r_i)=2i+1=\lambda(r_i)$.

\smallskip
\noindent
There is a final unreachable target case to consider, which we do so for completeness. Namely the case wherein $t\notin V_{\mathrm{rch}}$, necessitating $t$ is not equal to any $r_i$. Herein  $t$ has no incoming arc in $E'$ as all arcs in $E_0$ feature both of their endpoints
being in $V_{\mathrm{rch}}$, all arcs from attachments into the core enter some $r_i$,
and no arcs involved in the spine enters $t$. Hence $t$ is unreachable from $s'$ in $H$,
and therefore
$
\operatorname{dist}_H(s',t)=\infty=\lambda(t)$.

We have now proved
$
\operatorname{dist}_H(s',v)=\lambda(v) $
for every $v\in V'$. For every $d\ge 0$,
$
L_H(s',d)
=
\{v\in V':\operatorname{dist}_H(s',v)=d\}
=
\{v\in V':\lambda(v)=d\}$. This tells us that each BFS layer has at most one spine vertex and at most one non-spine vertex, whence the desired BFS-width claim immediately follows. 
\qedtri
\end{proof} 
\begin{claim}[Equisatisfiability of $\operatorname{Normalize}$]\label{clm:normalize-equisat}
$I=(G,s,t,\mathcal F)$ is a \textsc{YES}-instance
if and only if $\operatorname{Normalize}(I)$ is a
\textsc{YES}-instance.
\end{claim}

\begin{proof}

Let $I'=(G',s',t,\mathcal F')=\operatorname{Normalize}(I)$.

We first get out of the way a trivial case. 
Consider that $t\notin V_{\mathrm{rch}}$. Here, $I$ is trivially \textsc{NO}, as there is not even any path from $s$ to $t$ in $G$, let alone a safe path between them. $I'$ matches this in that it will also be \textsc{NO}, as it will not create any arc whose head is $t$.

Therefore, we are able to assume $t\in V_{\mathrm{rch}}$, which we do for the rest of this proof.

\noindent\textbf{($\Rightarrow$)}
Let $P$ be a path from $s$ to $t$ in $G$, and assume that it is safe. Each of the vertices it contains have membership in $V_{\mathrm{rch}}$, and every arc $P$ uses likewise is in the reachable part $E_{\mathrm{rch}}$. As $E_{\mathrm{rch}}\subseteq E_0\subseteq E'$, $P$ is so too a path in $G'$.

Now $
s\leadsto_{G_{\mathrm{rch}}} v$ holds for any  $v\in V_{\mathrm{rch}}\setminus\{s\}$.
Since $G_{\mathrm{rch}}$ is a DAG, we also have
$
v\not\leadsto_{G_{\mathrm{rch}}} s$.
Because $\triangleleft=\mathsf{LE}(G_{\mathrm{rch}})$ extends proper
reachability in $G_{\mathrm{rch}}$, it follows that
$
s\triangleleft v
$
for every $v\in V_{\mathrm{rch}}\setminus\{s\}$. Thus $s$ is first in the
$\triangleleft$-order. Since $(r_1,\ldots,r_q)$ is the reverse-$\triangleleft$
enumeration, we have
$
r_q=s$. Now define $P_{\mathrm{pref}} \coloneq 
s'\to p_1\to p_2\to \cdots \to p_{2q-1}\to w_{q}\to r_{q}(=s)$, a (directed) gadget prefix which uses only arcs of $E_{\mathrm{spine}}\cup E_{\mathrm{att}}\subseteq E'$.
Let $Q$ be the directed $s'$--$t$ path obtained by following $P_{\mathrm{pref}}$ and then appending the suffix of $P$ from $s$ to $t$ (starting with the successor of $s$ on $P$).

We now establish that $Q$ is a safe path in $G'$. $Q$ needs only avoid pairs which were either already in $\mathcal F_{\mathrm{rch}}$ or which are of the form  $\{p_{2i-1},w_i\}$ for some $i$ for which $r_i\neq s$. Any vertices which $Q$ could have that are elements of pairs in $\mathcal F_{\mathrm{rch}}$ must be in $P$, and $P$ is safe w.r.t. $\mathcal F_{\mathrm{rch}}$; thus $Q$ is safe w.r.t.
$\mathcal F_{\mathrm{rch}}$ as well. 
Also, the only vertex $Q$ uses of the form $w_i$ is $w_q$.
$\{p_{2q-1},w_q\}$ cannot be contained in $\mathcal F'$
because $r_q=s$. This means that $Q$ is safe.

\noindent\textbf{($\Leftarrow$)} Assume that $I'$ is satisfiable. Then there is a safe path from $s'$ to $t$ in $G'$; call it $Q$. We need to show that there exists a safe path from $s$ to $t$ in $G$.

First of all, $Q$ must contain vertices of $V_{\mathrm{rch}}$, as $t$ itself is in $V_{\mathrm{rch}}$. So we can talk about the first such vertex of $Q$ to be in $V_{\mathrm{rch}}$, which we denote by $r_k$. Only arcs of the form $(w_i, r_i)$ are able to have their head be in $V_{\mathrm{rch}}$ while their tail is outside $V_{\mathrm{rch}}$. So the vertex which immediately precedes $r_k$ in $Q$ has to be $w_k$. Also, $p_{2k-1}$ is the only vertex in $G'$ occurring in arc whose head is $w_k$, so $Q$ has not only $w_k$ but also $p_{2k-1}$.

Now if $r_k\neq s$, $Q$ cannot be safe because $\{p_{2k-1},w_k\} \in \mathcal F'$. So we have that $r_k = s$.

$Q$ cannot contain vertices in $
\{s'\}\cup P\cup W,
$ after $s$, which follows by construction. 
So the portion of $Q$ which goes from $s$ to $t$
contains only vertices in $V_{\mathrm{rch}}$,
so it is a path in $(V_{\mathrm{rch}},E_0)$.
 The arcs of this portion are also only in the arcs of $G$: if an arc in $ A_{\triangleleft}(\mathcal F_{\mathrm{rch}})$ was used here, $Q$ would unsafe, contradicting how we assumed that it is safe. So said portion of $Q$ is in fact a path in $G$. 

It is moreover safe with respect to $\mathcal F$. If we assume that both endpoints of $\{a,b\}$ for $\{a,b\} \in \mathcal F$ appear on this portion, we have $\{a,b\}\in\mathcal F_{\mathrm{rch}}$.
Since $Q$ is safe with respect to $\mathcal F'$, it is safe with respect to $\mathcal F_{\mathrm{rch}}$,
so the extracted $s$--$t$ path in $G$ is safe for $\mathcal F$ as well.
Thus $I$ is a \textsc{YES}-instance. \qedtri
\end{proof} 
Taken together, this proves Theorem~\ref{thm:pl-to-bfsw2}. \qedbox
\end{proof}

\subsection{\textsf{NP}-completeness at union BFS-width 2}

\begin{theorem}[\textsf{NP}-completeness on BFS-Width-$2$ Union DAGs]
\label{thm:bfsw2-union-npc}
\textsc{PAFP} is \textsf{NP}-complete even when restricted to instances
$I=(G=(V,E),s,t,\mathcal F)$ such that $G$ is a DAG and $\mathrm{bfsw}\bigl(G \cup \mathcal F,\, s\bigr)\ \le\ 2$.
\end{theorem}
\begin{proof}
Membership in \textsf{NP} is immediate. For \textsf{NP}-hardness, reduce from \textsc{PAFP} on DAGs, which is \textsf{NP}-complete \cite{gabow76,kolman09}. Given an arbitrary
instance $I=(G,s,t,\mathcal F)$ with $G$ a DAG, output
$
I' := \operatorname{Normalize}(I)$.
By Theorem~\ref{thm:pl-to-bfsw2}, $I$ is a \textsc{YES}-instance if and only if
$I'$ is a \textsc{YES}-instance, and the union digraph of $I'$ is a DAG of
BFS-width at most $2$. Since $\operatorname{Normalize}$ is polynomial-time
computable, this is a polynomial-time many-one reduction to the stated class. \qedbox
\end{proof}

\section{BFS Layers of the Input Graph}

\begin{corollary}[\textsf{NP}-completeness on BFS-width-$2$ input DAGs]\label{cor:bfsw2-input-hard}
\textsc{PAFP} is \textsf{NP}-complete when restricted to instances
$I=(G,s,t,\mathcal F)$ such that $G$ is a DAG and
$
\mathrm{bfsw}(G,s)\le 2 $.
\end{corollary}
 \begin{proof}[Proof sketch]
 By Claim~\ref{clm:normalize-dag-union}, $\operatorname{Normalize}$ produces \textsc{PAFP} instances whose union digraph and input graph are equal. Thus Theorem~\ref{thm:pl-to-bfsw2} implies that \linebreak$\mathrm{bfsw}(G',s') \le 2$ for the instance $\operatorname{Normalize}(I) = (G',s',t,\mathcal F')$, and that hard instances can be normalized in polynomial-time to equivalent instances whose input graph is a DAG of BFS-width at most 2. \qedbox
 \end{proof}

We next contrast the union-digraph result with the following: 
\begin{theorem}[Input BFS-width $2$ with few backward input arcs]
\label{thm:input-bfsw2-few-backward-fpt}
Let
$
I=(G=(V,E),s,t,\mathcal F)
$
be a \textsc{PAFP} instance where $G$ is a DAG, and assume that
$
\operatorname{bfsw}(G,s)\le 2$.
Let
$
k:=|B_G(s)|$
be the number of backward input arcs in $G$ from $s$.
Then $I$ can be decided in time
$
2^k\cdot |I|^{O(1)}
$.
In particular, for every fixed $k$, \textsc{PAFP} is polynomial-time decidable
on DAG instances of input BFS-width at most $2$ with at most $k$ backward
input arcs.
\end{theorem}

\begin{proof}
Let
$
R:=\{v\in V:\operatorname{dist}_G(s,v)<\infty\}$ and
$
G_R:=G[R]$.
If $t\notin R$, then there does not exist any directed $s$-$t$ path in $G$, and we
answer NO. Thus we can assume $t\in R$, which we now do. We also adopt the convention that all layers and distances in this proof are
taken in $G_R$, and we employ the notation
$
\lambda(v):=\operatorname{dist}_{G_R}(s,v)$ and $
L_d:=L_{G_R}(s,d)$.
Let
$
B:=B_G(s)
$
and let
$
E^+:=E(G_R)\setminus B$.
$G_R$ is a DAG because it is a subdigraph of $G$. For every arc
$(u,v)\in E(G_R)$, shortest-path distance gives
$
\lambda(v)\le \lambda(u)+1$.
Therefore, every arc of $E^+$ either is between vertices which are in the same BFS layer or is between vertices which differ by the depth of their BFS layers by at most 1.

We give an algorithm with the desired runtime properties to give an existence proof of the statement of this theorem. The first thing it does is enumerate all subsets $S$ of $B$. $S$ will contain all backward arcs in $G$ used by a safe path in it, if one exists. 

After we have $S$, any path that contains each of the arcs of $S$ can only have one order in which they could use them, which is because $G_R$ is acyclic and a topological ordering of it fixes this. Fix a topological ordering of $G_R$, and enumerate the arcs in $S$ according to the order in which their tails appear in said ordering as  $
e_1=(a_1,b_1),\ldots,e_r=(a_r,b_r)$. To handle $S = \emptyset$, it is just the case that $r = 0$ here.

Set $b_0 := s$ and $a_{r+1} := t$. If we have a path whose backward arcs are all of the arcs in $S$, if we were to remove all of the arcs in $S$ from this path, we would get $r+1$ path-segments, the start and end vertices of which are $(\alpha_j,\omega_j) := (b_j,a_{j+1})$ for $ 0 \le j \le r$. The only arcs used in each segment are those of $E^+$. Some of the segments may have zero length. In the case of $r = 0$, the candidate $s$--$t$ path is a single segment without any backward arcs. Also, 

The algorithm discards a working subset $S$ if for some $j$ it is $\lambda(\alpha_j) > \lambda(\omega_j)$. 

For a subset which remains after said deletion, segment $j$ contains a vertex a BFS layer from $\lambda(\alpha_j)$ through $\lambda(\omega_j)$ in order. Refer to each such visits of a vertex in a BFS layer by indexing them with $(j,d)$. Now as a segment with no backward arcs visits the BFS layers in the contiguous manner of a subpath aforementionedly described, and the bfs width is at most 2, said subpath is determined by its first and last vertices only. 

We use $\operatorname{Ent}(j,d,u)$ and $\operatorname{Exit}(j,d,u)$ to assert that $u \in L_d$ is chosen as respectively an entry or exist. A BFS layer with only one vertex necessitates that both choices are made, that is we have $\top$ for both expressions. 

For $L_d = \{u_d^0,u_d^1\}$, introduce Boolean variables
$e_{j,d}$ and $x_{j,d}$ for the entry and exit of this occurrence.
Use the literals
$
\bigl(\operatorname{Ent}(j,d,u_d^0),\operatorname{Ent}(j,d,u_d^1)\bigr)
    = (\neg e_{j,d},e_{j,d}),
\bigl(\operatorname{Exit}(j,d,u_d^0),\operatorname{Exit}(j,d,u_d^1)\bigr)
    = (\neg x_{j,d},x_{j,d}).$ For the clauses we have $\operatorname{Ent}(j,\lambda(\alpha_j),\alpha_j)$ and $\operatorname{Exit}(j,\lambda(\omega_j),\omega_j)$ where the idea of the first one is to fix the entry of the initial occurrence, and that of the second is to fix the exit of the final one. Another clause we have is $\neg\operatorname{Ent}(j,d,u) \vee \neg\operatorname{Exit}(j,d,v)$, with the idea of icnluding it being tha twe want to exclude each ordered choice of distinct $u,v \in L_d$
with $(u,v) \notin E^+$, for within an occurence $(j,d)$.

An arc joining two consecutive BFS layers the required arc joins an exit to an
entry. For occurrences $(j,d)$ and $(j,d+1)$ of the same segment,
for each $u \in L_d$ and $v \in L_{d+1}$ with $(u,v) \notin E^+$ we add the clause$
\neg\operatorname{Exit}(j,d,u) \vee
\neg\operatorname{Ent}(j,d+1,v)
$ to enforce this.

\smallskip
\noindent
\emph{Forbidden-pair constraints.} Now, we make sure that the path so created using arcs of $S$ and segments without backward arcs is safe.  At an occurrence $(j,d)$, a vertex $u\in L_d$
belongs to the $j$-segment without backward arcs encoded by the boolean variables when $u$ is at least one of the entry or exit. 

Thus this event is expressed by
$\operatorname{Ent}(j,d,u)\vee\operatorname{Exit}(j,d,u)$.

For every forbidden pair $\{a,b\}\in\mathcal F$, consider
every pair of occurrences $(j,d)$ and $(j',d')$ with
$a\in L_d$ and $b\in L_{d'}$, respectively. These
occurrences may coincide. To prevent them from using both
endpoints, add the clause $\neg P\vee\neg Q$ for every
choice of
\begin{equation}
\begin{aligned}
P&\in\{\operatorname{Ent}(j,d,a),
       \operatorname{Exit}(j,d,a)\},\\
Q&\in\{\operatorname{Ent}(j',d',b),
       \operatorname{Exit}(j',d',b)\}.
\end{aligned}
\label{eq:forbidden-occurrence-choices}
\end{equation}

The effect of the above clauses is to prevent both occurrences in a combination from using endpoints that are forbidden, as desired.

Finally, we define $\Phi_S$ as the conjunction of all
clauses we have defined in this proof so far.
After evaluating negations of Boolean constants, we
delete any clause with $\top$ and remove any occurrence
of $\bot$ from the remaining clauses.

If there is an empty clause anywhere, we discard the current subset $S$. Otherwise we solve $\Phi_S$ using a standard 2-sat algorithm, and we return \textsc{Yes} if and only if running the algorithm returns satisfiable for at least one subset $S\subseteq B$.

An empty clause makes $\Phi_S$ unsatisfiable; otherwise,
the resulting formula is in 2-CNF. Solve it using
\textsc{2-sat}, and accept the instance if $\Phi_S$
is satisfiable for at least one subset $S\subseteq B$.

We now argue that the algorithm about which we just talked is correct. Assume now that $G$ has a safe path from $s$ to $t$, call it $P$. Let $S$ be the set of backward arcs in $P$.  The order of the arcs in $S$ must be consistent with the topological ordering fixed by us earlier in this proof. Removing these arcs gives us the segments $(r+1)$ subpaths in which only  $E^+$ is used for the arcs. The endpoint-layer test does not reject $S$ as the BFS depths of the layers here do not decrease. 

For each occurrence $(j,d)$, we designate the first and last vertices of segment $j$ in $L_d$ as being respectively the entry and exit vertices. any distinct choices for entry and exit vertices must be joined by an arc of $E^+$ which does not go across distinct bfs layers. Occurences which are in consecutive BFS layers similarly satisfy that they are connected by arcs of $P$. Thus the clauses for within and across the bfs layers are satisfied, and the unit clauses are satisfied by the segment endpoints. The forbidden-pair clauses are also satisfied because each selected vertex is part of $P$, and $P$ is safe. Thus $\Phi_S$ is satisfiable.

Conversely, assume that  $\Phi_S$ is satisfiable. Whatever chosen for the entry and exit of an occurence $(j,d)$ is descriptive of either a single vertex or arc of $E^+$ (the latter of which things is the case from clauses within the same BFS layer). Occurrences in consecutive BFS layers must be connected by valid arcs because of the clauses for across the BFS layer. Finally, applying the endpoint clauses gives us that the choices these choices define
a  walk $W_j$ from $\alpha_j$ to $\omega_j$
using only elements in $E^+$.

Now because $\omega_{j-1}=a_j$ and $\alpha_j=b_j$, we can, for $1\le j\le r$, insert the arc $e_j=(a_j,b_j)$ considered by a subset $S$ considered by the algorithm  between $W_{j-1}$ and $W_j$.
With this, we get a walk from $s$ to $t$ in $G_R$ the backward arcs of which are $S$. This walk is also a directed path due to the acyclity of $G_R$. Every vertex on this path occurs as some entry or exit; in particular, the endpoints of each arc in the $S$ currently considered by the algorithm are the exit and entry vertices of its
adjacent segments. If the path contained both endpoints
of some $\{a,b\}\in\mathcal F$, then two expressions
of the forms in \eqref{eq:forbidden-occurrence-choices}
would both be true: but this would violate the clause
$\neg P\vee\neg Q$. Hence the path is safe.

\smallskip
\noindent
\emph{Running time.}
For a  $S\subseteq B$, there are at most
$k+1$ segments and per segment at most $|V|$ occurences in a BFS layer.  Each occurrence gives at most two
Boolean variables and a constant number of the clauses involved in endpoint and adjacency conditions, which so far gives us $O((k+1)|V|)$ things. 

Each reachable vertex belongs to exactly one BFS layer and per segment can occur only at most 1 time. As such, there are at most
$4(k+1)^2$ clauses from each forbidden pair, giving
$O(|\mathcal F|(k+1)^2)$ forbidden-pair clauses
overall. Since $k\le|E|$, constructing and solving
$\Phi_S$ takes $|I|^{O(1)}$ time for each $S$. Finally, that we have to enumerate the subsets of $B$, taking $O(2^k)$ time, means the overall running time is $2^k\cdot|I|^{O(1)}$, as desired.

\end{proof}

At input BFS-width $3$, hardness already occurs in DAGs with
no backward input arcs.

\begin{lemma}[Layered form of the Gabow--Maheshwari--Osterweil reduction]
\label{lem:gmo-layered-structure}
There is a polynomial-time reduction from \textsc{3-sat} to
\textsc{PAFP} that maps a formula with $m$ clauses of three
literal occurrences each to an instance $I = (G,s,t,\mathcal{F})$.
The input digraph $G$ is a DAG with vertex partition
$L_0,\ldots,L_{m+1}$ satisfying
\[
L_0 = \{s\}, \qquad |L_i| = 3 \quad (1 \le i \le m),
\qquad L_{m+1} = \{t\}.
\]
Its arcs consist of all arcs from $L_i$ to $L_{i+1}$ for
$0 \le i \le m$. The formula is satisfiable if and only if $I$
has a safe $s$-$t$ path.
\end{lemma}

\begin{proof} This is immediate from the proof of Lemma 2 in~\cite{gabow76} by Gabow, Maheshwari, and
Osterweil. \qedbox
\end{proof}
\begin{proposition}[Input BFS-width $3$ with no backward arcs is hard]
\label{prop:input-width3-no-backward-hard}
\textsc{PAFP} is $\mathsf{NP}$-complete on DAG instances
$I=(G,s,t,\mathcal F)$ satisfying
$
\operatorname{bfsw}(G,s)\le 3$ and $B_G(s)=\emptyset$.
\end{proposition}

\begin{proof}
Membership in $\mathsf{NP}$ is immediate. For hardness, use the instances
from Lemma~\ref{lem:gmo-layered-structure}. Since every arc goes from
$L_i$ to $L_{i+1}$, the directed BFS layers from $s$ are exactly
$
L_0,L_1,\ldots,L_{m+1}$.
Thus every BFS layer has size at most $3$, so
$
\operatorname{bfsw}(G,s)\le 3$.
Moreover, every input arc advances one BFS layer, so
$
B_G(s)=\emptyset$.
The same instances are $\mathsf{NP}$-hard by
Lemma~\ref{lem:gmo-layered-structure}.\qedbox
\end{proof}

\section{Exact-Length Layers of the Input Graph}

The first result in this section is that Exact-length width 2 is polynomial-time decidable.


\begin{theorem}[\textsc{PAFP} on DAGs of exact-length width at most $2$]\label{thm:exactwidth2}
Let $I=(G=(V,E),s,t,\mathcal F)$ be a \textsc{PAFP} instance where $G$ is a DAG.
Assume $\mathrm{elw}(G,s)\le 2$. Then $I$ can be decided in deterministic polynomial time.
\end{theorem}
\begin{remark}[the promise still allows many paths]\label{rem:elw-exp}
The condition $\mathrm{elw}(G,s)\le 2$ does \emph{not} imply that $G$ has few $s$--$t$ paths.
For example, let $G$ be a properly layered DAG with layers
$A_0=\{s\}$, $A_d=\{u_d,v_d\}$ for $d=1,\ldots,\ell-1$, and $A_\ell=\{t\}$, and include all arcs
from each layer to the next (i.e., every vertex in $A_d$ has arcs to all vertices in $A_{d+1}$).
Then $E_G(s,d)=A_d$ for all $d$, so $\mathrm{elw}(G,s)=2$, but $G$ has $2^{\ell-1}=2^{\Theta(|V|)}$
distinct directed $s$--$t$ paths.
Thus Theorem~\ref{thm:exactwidth2} gives polynomial-time solvability even in instances with an exponential amount of paths from $s$ to $t$ as well as an arbitrary set of forbidden pairs (in contradistinction to known tractability regimes for \textsc{PAFP}).

\end{remark}
\begin{remark}\label{rem:width-contrast}
For every DAG $G$, every $d\ge 0$, and every root $s$, one has
$L_G(s,d)\subseteq E_G(s,d)$, and hence
$\mathrm{bfsw}(G,s)\le \mathrm{elw}(G,s)$.
Thus the restriction $\mathrm{elw}(G,s)\le 2$ is stricter than
$\mathrm{bfsw}(G,s)\le 2$. Our results show that this stronger local promise gives polynomial solveability, which is tight in the sense that increasing the exact-length width by $1$ to get  $3$ gives hardness, whereas BFS-width $2$ does not suffice
for tractability.
\end{remark}
\begin{proof}[Proof of Theorem~\ref{thm:exactwidth2}]
Fix an instance $I=(G=(V,E),s,t,\mathcal F)$ where $G$ is a DAG and $\mathrm{elw}(G,s)\le 2$.
Let $n:=|V|$. Since $G$ is acyclic, every directed path is simple, so its length is at most $n-1$. So in this proof we only consider lengths $\ell\in\{0,1,\ldots,n-1\}$.

We first deal with how to compute the exact-length layers, which the algorithm we give needs to do. 
For each $d \in \{0,\ldots,n-1\}$ define $D_d \coloneq E_G(s,d)$.
We compute $D_0,\ldots,D_{n-1}$ by the DP
\begin{equation}
D_0 := \{s\},\qquad
D_d := \{\,v\in V \mid \exists (u,v)\in E \text{ with } u\in D_{d-1}\,\}\quad (d\ge 1).
\end{equation} That the DP correctly computes the exact-length layers is immediate by induction on $d$., and also as far as its runtime, evaluating each $D_d$ by scanning all arcs once takes $O(|E|)$ time per $d$, hence $O(n|E|)$ total time. By assumption $\mathrm{elw}(G,s)\le 2$, we have $|D_d|\le 2$ for all $d\ge 0$.

The idea is to encode any safe path of length $\ell$
as a \textsc{2-sat} instance, and check for such paths
by querying \textsc{2-sat} instances. We consider only
lengths $\ell\in\{0,\ldots,n-1\}$ for which $t\in D_\ell$,
skipping all other lengths. The algorithm accepts if
at least one of these instances is satisfiable.

Any directed $s$--$t$ path of length exactly $\ell$ can of course be represented as a sequence
$s=X_0 \to X_1 \to \cdots \to X_\ell=t$,
where $X_d$ is the vertex used at position $d$ in the path.
It must be that $X_d \in D_d$ for each $d$, and consecutive positions must satisfy $(X_d,X_{d+1})\in E$.
The forbidden pairs impose additional constraints of the form “it is not allowed that
$X_i=a$ and $X_j=b$ simultaneously.”
Under the assumption $|E_G(s,d)|\le 2$, each position $d$ has at most two possible vertices: so this can be represented by some boolean variable.  
Moreover, each constraint forbids \emph{one} joint assignment of \emph{two} positions, so it is naturally expressible
as a 2-CNF clause. Thus what we will do is reduce PAFP under these assumptions to \textsc{2-sat}. 

\noindent
We now specify the variables in the \textsc{2-sat} instance to which we reduce. 
If $|D_d|=2$, fix a deterministic ordering $D_d=\{v_d^0,v_d^1\}$ and introduce a Boolean variable $x_d$ with the intended
meaning of $x_d=\textsc{false}$ chooses $v_d^0$ as the $d$th vertex on a length $\ell$ path, and $x_d=\textsc{true}$ chooses $v_d^1$.
If $|D_d|=1$, write $D_d=\{v_d\}$; then position $d$ is forced and we introduce no variable.

Then we make the literals such that they encode that “position $d$ selects vertex $u$.” Namely, define the literal (or constant) $\mathrm{Sel}(d,u)$ for each $u\in D_d$ by
\begin{equation}
\mathrm{Sel}(d,u)\;=\;
\begin{cases}
\neg x_d & \text{if }|D_d|=2\text{ and }u=v_d^0,\\
x_d      & \text{if }|D_d|=2\text{ and }u=v_d^1,\\
\top     & \text{if }|D_d|=1\text{ and }D_d=\{u\}.
\end{cases}
\end{equation}
Thus $\mathrm{Sel}(d,u)$ evaluates to \textsc{true} exactly when the vertex at position $d$ selected by the variable assignment is $u$. 

We now go on to the clauses. Every single clause in the \textsc{2-sat} instance to which we reduce, with the exception of that mandating that $t$ is an endpoint, expresses that $X_i=a$ and $X_j=b$ cannot both happen, which is for the forbidden pairs, to prevent vertices between which there is no arc from being chosen as consecutive positions in the path, and making sure that the path ends at $t$. We spell out these clauses now.

We enforce that position $\ell$ is $t$ by adding the (unit) clause $\mathrm{Sel}(\ell,t)$.
(If $|D_\ell|=1$ this is $\top$ and does nothing; otherwise it fixes $x_\ell$ appropriately.)

For each $d\in\{0,\ldots,\ell-1\}$ and each $u\in D_d$, $v\in D_{d+1}$ with $(u,v)\notin E$,
we forbid choosing $u$ at position $d$ and $v$ at position $d+1$ by adding $\neg \mathrm{Sel}(d,u)\ \lor\ \neg \mathrm{Sel}(d+1,v)$.
Since $|D_d|,|D_{d+1}|\le 2$, there are at most $4$ such checks per $d$.

Finally, 
for each forbidden pair $\{a,b\}\in\mathcal F$ and each pair of \emph{distinct} positions $i\neq j$
with $a\in D_i$ and $b\in D_j$, we forbid simultaneously selecting $a$ and $b$ by adding
$\neg \mathrm{Sel}(i,a)\ \lor\ \neg \mathrm{Sel}(j,b)$.
(If a vertex is forced at some position, $\mathrm{Sel}$ becomes $\top$, and this correctly collapses
to a unit clause forbidding the other endpoint at the other position.)

Let $\Phi_\ell$ be the conjunction of all clauses above. The only Boolean variables occurring in $\Phi_\ell$ are the variables $x_d$ introduced for indices $d\in\{0,\ldots,\ell\}$ with $|D_d|=2$.

Our first step is to simplify by substituting $\top$ and $\bot$ as Boolean constants and then deleting any clauses containing $\top$,
removing any occurrence of $\bot$ from a clause, 
and then discarding the current length $\ell$ if
some clause becomes empty in this process.
The resulting equivalent formula contains no constants and is a standard \textsc{2-sat} instance. 

Because ti si clearer to do so, in our analysis we work with the original unsimplified formula.

We now show that $\Phi_\ell$ is satisfiable if and only if there exists a path from $s$ to $t$ of length $\ell$ that is safe.

\emph{($\Rightarrow$)}
Given a satisfying assignment to the variables of $\Phi_\ell$, define $X_d$ as follows:
if $D_d=\{v_d\}$ then $X_d=v_d$, and if $D_d=\{v_d^0,v_d^1\}$ then $X_d=v_d^0$ when $x_d=\textsc{false}$
and $X_d=v_d^1$ when $x_d=\textsc{true}$.
$X_\ell=t$ follows from the clauses, and since $D_0=\{s\}$ we have $X_0=s$. 

For every $d\in\{0,\ldots,\ell-1\}$, that  $(X_d,X_{d+1})$ in $E$ holds as otherwise the unique
clause forbidding that pair would be violated. Hence
$s=X_0 \to X_1 \to \cdots \to X_\ell=t$
is a walk of length $\ell$ in $G$, and in a DAG walks are paths so this is also a path.

Finally, consider any forbidden pair $\{a,b\}\in\mathcal F$.
If the path contained both $a$ and $b$, then there would exist distinct positions $i\neq j$ such that $X_i=a$ and $X_j=b$, violating the clauses: so the path is safe.

\noindent
\emph{($\Leftarrow$)}
Let $P \coloneq s=u_0 \to u_1 \to \cdots \to u_\ell=t$
be a safe path from $s$ to $t$ of length $\ell$.
Then $u_d\in D_d$ for each $d$.
Now for each $d$ satisfying $|D_d|=2$ and $D_d=\{v_d^0,v_d^1\}$, make an assignment to $x_d$ such that $X_d$ is $u_d$.
This assignment satisfies the clause involving what we want the desired endpoints of the path to be as it has $u_\ell=t$, and it satisfies clauses governing adjacency between path positions (it has$(u_d,u_{d+1})\in E$), clauses for the forbidden pairs satisfied also due to the safeness of $P$. Thus $\Phi_\ell$ is satisfiable.

We still need to verify that the running time is polynomial.  Let $m:=|E|$ and $f:=|\mathcal F|$. Computing $D_0,\ldots,D_{n-1}$ takes $O(nm)$ time. Now fix $\ell$. Since $|D_d|\le 2$, constraint \textbf{(ii)} inspects at most $|D_d|\cdot|D_{d+1}|\le 4$ pairs per $d<\ell$,
hence contributes $O(\ell)$ clauses. For \textbf{(iii)}, for each vertex $v$ let
$\mathrm{occ}_\ell(v):=\{d\in\{0,\ldots,\ell\}\mid v\in D_d\}$; the sets $\mathrm{occ}_\ell(v)$ can be built in $O(\ell)$ time by scanning $D_0,\ldots,D_\ell$, since
$\sum_{d=0}^\ell |D_d|\le 2(\ell+1)$. Each $\{a,b\}\in\mathcal F$ contributes at most
$|\mathrm{occ}_\ell(a)|\,|\mathrm{occ}_\ell(b)|\le(\ell+1)^2$ clauses, so $|\Phi_\ell|=O(\ell+f\ell^2)$.
Building $\Phi_\ell$ and solving it via \textsc{2-sat} both take $O(|\Phi_\ell|)$ time, hence $O(\ell+f\ell^2)$. Summing over $\ell\in\{0,\ldots,n-1\}$ yields
$O(nm)+\sum_{\ell=0}^{n-1}O(\ell+f\ell^2)=O(nm+n^2+fn^3)$, which is polynomial. \qedbox
\end{proof}

We also show that exact-length width 3 is \textsf{NP}-complete.

\begin{proposition}[Exact-length width $3$ is already hard]
\label{prop:elw3-hard}
\textsc{PAFP} is \textsf{NP}-complete on DAGs $G$ with
$\mathrm{elw}(G,s)\le 3$. \end{proposition} 
\begin{proof} Membership in $\mathsf{NP}$ is immediate. For hardness, again use the
instances from Lemma~\ref{lem:gmo-layered-structure}. Since every arc goes
from $L_i$ to $L_{i+1}$, every directed path from $s$ reaches layer
$L_i$ after exactly $i$ steps. Hence
$
E_G(s,i)=L_i
$
for $i=0,\ldots,m+1$, and $E_G(s,i)=\emptyset$ for all other $i$.
Therefore
$
\operatorname{elw}(G,s)\le 3$.
The hardness follows from Lemma~\ref{lem:gmo-layered-structure}.\qedbox
\end{proof}

\section*{Acknowledgments}
The author thanks Henning Fernau for encouraging him to think about the maximum backward arcs crossing a BFS-layer rather than the absolute number of backward arcs itself -- following a discussion with Henning, after which the author proved that this more subtle parameter, with bounded BFS-width, yields FPT, and that one could also bound the maximum number of heads crossing a BFS-layer instead -- Henning then suggested that one could look at the minimum of the heads and tails crossing a bfs layer, while taking the maximum over all layers, which caused the author to state and prove Theorem~\ref{thm:bfs-backward-fpt}.

\bibliographystyle{splncs04}
\bibliography{BFSnormalform}

\appendix

\section{MSO Formulation for Bounded-Treewidth}
\label{app:mso}
We now spell out the bounded-treewidth subroutine used in the last step of the proof of Theorem~\ref{thm:bfs-backward-fpt}.
Let
$
k:=\operatorname{tw}(U[R_G])$.
Consider the finite relational structure
$
\mathcal A_{I_G}:=(R_G,\mathsf{Arc},\mathsf{Forb},s,t)$,
where the universe is $R_G$, where
$\mathsf{Arc}(x,y)$ holds precisely when $(x,y)\in E(G[R_G])$, and
where $\mathsf{Forb}(x,y)$ holds precisely when
$\{x,y\}\in \mathcal F_G$. Thus $\mathsf{Arc}$ is a directed binary
relation, while $\mathsf{Forb}$ is a symmetric binary relation encoding the
forbidden pairs. The Gaifman graph of $\mathcal A_{I_G}$ is
$U[R_G]$, as two distinct vertices occur together in some tuple of
$\mathsf{Arc}$ or $\mathsf{Forb}$ if and only if they are adjacent in the
underlying undirected graph of $G[R_G]$ or form a forbidden pair in
$\mathcal F_G$. $s$ and $t$ can be treated as either constants or unary singleton predicates, in neither of which cases they increase the Gaifman graph.

The existence of a safe directed $s$-$t$ path is expressible by a fixed
monadic second-order formula over this structure. Namely, let
$\varphi_{\mathrm{PAFP}}(s,t)$ be the formula
\[
\begin{aligned}
\exists Q\,\bigl(&Q(s)\land Q(t)\\
&{}\land
\forall x\forall y\,
\bigl((Q(x)\land Q(y))\implies \neg \mathsf{Forb}(x,y)\bigr)\\
&{}\land
\forall X\,\bigl(
  [\,\forall z\,(X(z)\implies Q(z))\land X(s)\land \neg X(t)\,]\\
&\hspace{3.6cm}\implies
  \exists x\exists y\,
  (X(x)\land Q(y)\land \neg X(y)\land \mathsf{Arc}(x,y))
\bigr)\bigr).
\end{aligned}
\]
Here $Q$ is intended to be a safe set of vertices containing $s$ and
$t$. The meaning of the second line is that $Q$ doesn't have a forbidden pair. The last two lines are a standard way to formulate reachability (directed-reachability) from $s$ to $t$ inside the subdigraph induced by $Q$, using how : every subset of $Q$ that
contains $s$ but not $t$ has an outgoing $\mathsf{Arc}$-edge to
$Q\setminus X$.

We claim that
$
\mathcal A_{I_G}\models \varphi_{\mathrm{PAFP}}(s,t) $
if and only if $I_G$ is a \textsc{YES}-instance. Indeed, if $I_G$ has a
safe directed $s$-$t$ path, take $Q$ to be the vertex set of this path. The second line of the formula is from the appropriate safety condition.  And for every
$X\subseteq Q$ with $s\in X$ and $t\notin X$, the first vertex of the path
outside $X$ has its predecessor inside $X$, giving the outgoing arc
required by the cut condition. Conversely, suppose the formula holds for some
set $Q$. Let $R\subseteq Q$ be the set of vertices reachable from $s$ in
the directed subgraph $G[R_G][Q]$. If $t\notin R$, then $R$ is a subset
of $Q$ containing $s$ but not $t$, and by definition of reachability
there is no arc from $R$ to $Q\setminus R$, contradicting the cut condition.
Thus $t$ is reachable from $s$ in $G[R_G][Q]$. Taking a vertex-simple
directed $s$-$t$ path inside this subgraph gives a directed $s$-$t$
path whose vertices are all contained in $Q$, and hence the path is safe
because $Q$ contains no forbidden pair.

Therefore \textsc{PAFP} on $I_G$ is an MSO model-checking problem on a
relational structure whose Gaifman graph is $U[R_G]$. Using the  MSO
model-checking theorem for bounded-treewidth relational structures
\cite[Sec.~3]{arnborgLS91}, also featured in the bounded-treewidth theorems used
for forbidden-pairs path problems by Bodlaender--Jansen--Kratsch
\cite[Prop.~6]{bodlaender13}, this model-checking problem can be decided in
time
$
g(k)\cdot |I|^{O(1)}$
for some computable function $g$. In the present proof, we do not need to
compute the exact value of $k$ as the path decomposition constructed in the proof of Theorem~\ref{thm:bfs-backward-fpt}, restricted to the vertex set $R_G$, is a path decomposition of $U[R_G]$ of width at most
$2b+2p-1$. Hence
$
k=\operatorname{tw}(U[R_G])
\le \operatorname{pw}(U[R_G])
\le 2b+2p-1$.
It follows that $I_G$, and therefore also the original instance $I$, can be
decided in time
$
f(b+p)\cdot |I|^{O(1)}$
for some computable function $f$.

\end{document}